\documentclass[10pt,tightenlines,eqsecnum,floats,aps,amsmath,
amssymb,nofootinbib,prd%,showpacs
]{revtex4}

\usepackage{amsmath,amsfonts,amssymb,amsthm}
\usepackage{vmargin}
\usepackage{colordvi}

\newcommand{\rd}{{\rm d}}
\newcommand{\ciag}{{\cal C}}
\newcommand{\I}{{\cal I}}

\newcommand{\ei}{i_e}
\newcommand{\eo}{o_e}
\newcommand{\tei}{\tilde{\imath}_e}
\newcommand{\teo}{\tilde{o}_e}
\newcommand{\tfi}{\tilde{\imath}_f}
\newcommand{\tfo}{\tilde{o}_f}
\newcommand{\ti}{{\tilde{\imath}}}
\newcommand{\tj}{{\tilde{\jmath}}}
\newcommand{\tk}{{\tilde{k}}}

\DeclareMathOperator{\tr}{Tr}
\DeclareMathOperator{\Inv}{Inv}
\DeclareMathOperator{\Val}{Val}

\newtheorem{lm}{Lemma}
\newtheorem{thm}{Theorem}
\newtheorem{col}{Corollary}
\newtheorem{fct}{Fact}
\newtheorem{df}{Definition}

\begin{document}

\title{All 3-edge-connected relativistic BC and EPRL spin-networks are integrable}

\author{Wojciech Kami\'nski}
\email{wkaminsk@fuw.edu.pl}
\affiliation{Instytut Fizyki Teoretycznej, Uniwersytet Warszawski, ul. Ho{\.z}a 69, 00-681 Warszawa (Warsaw), Polska (Poland)}

%\pacs{04.60.Pp}

\begin{abstract}
We prove statement conjectured in \cite{BB} %[Baez and Barrett:2001]
that every 3-edge-connected $SL(2,{\mathbb C})$ %SL(2,C)
spin-network with invariants of certain class is integrable. It means that the regularized evaluation (defined by a suitable integral) of such a spin-network is finite. Our proof is quite general. It is valid for relativistic spin-networks of Barrett and Crane as well as for spin-networks with the Engle-Pereira-Rovelli-Livine intertwiners and for some generalization of both. The result interesting from the group representation point of view opens also a possibility of defining vertex amplitudes for Spin-Foam models based on non-simplicial decompositions.
\end{abstract}
\maketitle

\section{Introduction}

Let $G$ be a compact group. A \emph{spin-network} is an object associated
with a given graph $E$. The graph $E$ consists of nodes (denoted by $E_0$) and oriented links connecting them (denoted by $E_1$). With each link $e$ of the graph we associate an irreducible, unitary representation of $G$ and with each node $i$ an invariant in the tensor product of representations (and dual representations) associated with links outgoing from $i$ (respectively, links incoming to the node $i$). The evaluation of the spin-network is the result of suitable contractions of the invariants. Explicitly, we contract two indices of the given invariants if they correspond to two ends of a link (see~\cite{Penrose} and also~\cite{Baez}).

In the case of $SL(2,{\mathbb C})$, a non-compact group, one encounters several problems with this definition. First of all, invariants are necessarily unbounded objects (see Appendix \ref{inv-sec}). Moreover, even if we are able to define contraction for such objects, the result will be almost always infinite \cite{BC}. The method how to define the evaluation for certain class of such spin-networks was proposed in \cite{BC}. In principle, the result is given by the integral over several copies of $SL(2,{\mathbb C})$ (see Sec. \ref{sec-tech}). We will say that \emph{the spin-network is integrable} if the corresponding integral is absolutely convergent (integral of modulus of the function is finite).

It can be expected from representation-theoretic consideration that this does not work for spin-networks that are not 3-edge-connected (definition \ref{three-con-df}). However, there is a conjecture (stated in \cite{BB}) that the procedure works for all other graphs. Some partial results in this direction have been obtained in \cite{Cherrington}. In this paper we will prove the conjecture in full generality for spin-networks labelled by the Barrett-Crane \cite{BC,Reisenberger} and the Engle-Pereira-Livine-Rovelli \cite{EPRL} (see also \cite{Barrett}) intertwiners. However, our result is valid also for some natural generalization of both (see Eq.~\ref{F-map}).

These invariants are of special importance because the evaluations of the corresponding spin-networks are the vertex amplitudes in the Spin-Foam models \cite{Oeckl,kkl1}. In the case of models based on simplicial decompositions, the vertex graphs are $1$ skeletons of the $4$-simplices. In this case, the finiteness of the evaluation has been proved in \cite{BB,Cherrington,Engle}. Whether these results extend to more general decomposition has been unclear till now.  Our result opens the possibility of defining Lorentzian Spin-Foams in the general framework of cellular decompositions \cite{Oeckl,kkl1}. Our proof should be compared with  methods of Feynman diagrams and also with \cite{kontsevich}, where problem of convergent of a similar integral was solved by a method of suitable compactification of the domain of integration.

Let us mention at the end that in the case of the BC model there are results concerning finiteness of the whole state sum in the case of simplicial decomposition \cite{Rovelli,Cherrington}. We will not consider this issue in the current paper.

\subsection{The evaluation of the spin-network}\label{def-spin-net}

Let us describe a method to regularize the evaluation proposed in \cite{BC,BB}.
In the tensor product $V_{k_1,\rho_1}\otimes\cdots\otimes V_{k_n,\rho_n}$ of $SL(2,{\mathbb C})$ irreducible, unitary representations $V_{k_i,\rho_i}$ from the principal series (see~\cite{Ruhl,Naimark}) we have a subspace
\[
 {\cal F}(V_{k_1,\rho_1}\otimes\cdots\otimes V_{k_n,\rho_n})=\sum_{\{j_1,\ldots,j_n\}} V_{k_1,\rho_1}(j_1)\otimes\cdots\otimes V_{k_n,\rho_n}(j_n),
\]
where $V_{k_i,\rho_i}(j_i)$ is the $SU(2)$ subspace of $V_{k_i,\rho_i}$ of spin $j_i$, and $\sum$ is algebraic direct sum. We will denote by $\cal F$ such vectors in any tensor product of $SL(2,{\mathbb C})$ irreducible representations of the principal series.

If $n\geq 3$ then the formula\footnote{Our normalization of the measure is given by Eq. \ref{dg1} and \ref{dg2}.}
\begin{equation}\label{F-map}
 {\cal F}(V_{k_1,\rho_1}\otimes\cdots\otimes V_{k_n,\rho_n})\ni \iota\rightarrow \int_{SL(2,{\mathbb C})} \rd g\: T_{k_1,\rho_1}(g)\otimes\cdots\otimes T_{k_n,\rho_n}(g) \iota
\end{equation}
gives an invariant \cite{BC}, defined by duality as a functional on ${\cal F}(V_{k_1,\rho_1}\otimes\cdots\otimes V_{k_n,\rho_n})$ (see Appendix \ref{inv-sec}).
We can use $\iota$ to label invariants obtained in this way. Let us note, that this labelling is not one-to-one. Moreover, one can restrict label $\iota$ to the space of $SU(2)$ invariants. We will denote by $[\iota]\in[{\cal F}]$ the class of equivalence of labels that give the same invariant as $\iota$.

For a spin-network with the invariants from the set described above we can define \emph{the evaluation} by the method proposed in \cite{BC,BB}. 
\begin{enumerate}
\item Firstly, we consider the graph with links labelled by the representations of the principal series $(k_e,\rho_e)$ and nodes labelled by elements of ${\cal F}$ in the tensor product of the representations associated with the links meeting in the node (as described in the case of compact groups). Notice that the dual representation 
%(in the case of principal series) 
is isomorphic to the original one, thus it is also in the principal series. However, the isomorphism is not natural and that is why we would like to avoid it in the definition.
\item For every node $i\in E_0$ of the graph except one, from the label $\iota_i\in{\cal F}$ we form the invariant using equation \ref{F-map}. In order to compute the evaluation we contract the invariants obtained in this way with the $\iota_0$ in the chosen node (see Eq. \ref{int-tech} for an exact definition). It follows from Fact \ref{thm-label} (see Sec. \ref{sec-statement}) that the procedure is independent of the choice of specific labels for the given class in $[{\cal F}]$ of the invariants.
\item The evaluation of the spin-network labelled by invariants is equal to the evaluation of a spin-network labelled by \emph{any} representants of these invariants from $\cal F$.
\end{enumerate}
The evaluation depends on the choice of the Haar measure. If we change the measure
\[
 \rd g\rightarrow a\rd g,\ \ a>0\ ,
\]
then in order to obtain the same invariants we need to multiply labels $\iota_i$ by $a^{-1}$. Thus, the evaluation of the spin-network labelled by invariants also changes by factor $a^{-1}$. The situation is a bit less satisfactory than in the case of compact groups, where the choice of probabilistic Haar measure fixes the evaluation completely.

\subsection{Statement of the results}\label{sec-statement}

Let us remind definition.
\begin{df}\label{three-con-df}
Spin-network $E$ is 3-edge-connected if for every division of its node set $E_0$ into two disjoint and nonempty sets $E_0^1$ and $E_0^2$, there exist at least three links connecting nodes of $E_0^1$ with these of $E_0^2$.
\end{df}
The goal of our paper is to prove the theorem stated as a conjecture in \cite{BB} in a bit more general form:
\begin{thm}\label{fin-thm}
Every 3-edge-connected graph with labels of nodes from $[\cal F]$ is integrable.
\end{thm}
From the theorem one can deduce some special results:

The Barrett-Crane intertwiner $\iota_{BC}\in\Inv V_{0,\rho_1}\otimes\cdots\otimes V_{0,\rho_n}
$ \cite{BC,Reisenberger} is associated with the $\iota\in {\cal F}(V_{0,\rho_1} \otimes\cdots\otimes V_{0,\rho_n})$ of the form
\[
 \iota=v_{0,\rho_1}\otimes\cdots \otimes v_{0,\rho_n},
\]
where $v_{0,\rho_i}$ in each representation $V_{0,\rho_i}$ denotes a unique up to phase $SU(2)$ invariant normalized vector. The relativistic BC spin-network is a spin-network whose links are labelled by representations $(0,\rho_e)$ and nodes are labelled by $\iota_{BC}$.

\begin{col}\label{BC-fin-thm}
Every 3-edge-connected relativistic BC spin-network is integrable.
\end{col}

A similar situation occurs in the case of the EPRL models.
The map of \cite{EPRL} (see also \cite{Barrett}) is a map from the $SU(2)$ invariants into $SL(2,{\mathbb C})$ invariants for each node of the graph. The EPRL intertwiner map
\[
 \iota_{EPRL} \colon\Inv V_{k_1}\otimes\cdots\otimes V_{k_n}\rightarrow \Inv V_{k_1,\rho_1}\otimes\cdots\otimes V_{k_n,\rho_n}
\]
is the composition of the map into $\cal F$
\begin{equation}\label{EPRL-map}
V_{k_1}\otimes\cdots\otimes V_{k_n}\xrightarrow{\phi_{k_1}\otimes\cdots\otimes\phi_{k_n}} 
{\cal F}(V_{k_1,\rho_1}\otimes\cdots\otimes V_{k_n,\rho_n})
\end{equation}
restricted to $\Inv V_{k_1}\otimes\cdots\otimes V_{k_n}$ with \ref{F-map}\footnote{In fact, additional restriction on the numbers $k_e,\rho_e$ is imposed (see \cite{EPRL}) and usually nodes are assumed to be 4-valent (but see \cite{kkl1}).}. The map $\phi_{k}\colon V_k\rightarrow V_{k,\rho}$ is a unique (up to unitary equivalence) embedding of the  $SU(2)$ representation $V_k$ into $V_{k,\rho}$. The relativistic EPRL spin-network is a spin-network whose links and nodes are labelled by representations $(k_e,\rho_e)$ from the principal series and, respectively, by images of the EPRL maps. We have
\begin{col}\label{EPRL-fin-thm}
Every 3-edge-connected relativistic EPRL spin-network is integrable.
\end{col}
Let us also state
\begin{fct}\label{thm-label}
The evaluation of the 3-edge-connected graph $E$ depends only on the classes of equivalence $[\iota_i]$, $i\in E_0$ of the labels of the nodes of the graph.
\end{fct}
We will prove it in Appendix \ref{inv-sec}. This fact allows us to regard the evaluation as a procedure on spin-networks labelled by $SL(2,{\mathbb C})$ invariants described in Eq.~\ref{F-map}.

\section{Technical prelude}
\label{sec-tech}

In this section we derive Theorem \ref{fin-thm} from technical Theorem \ref{tech}.
We will use the following notation
\begin{itemize}
\item $E_0$ is the set of nodes of the graph labelled by elements of $\cal F$. Nodes will be denoted by $i,j,k$,
\item $E_1$ is the set of oriented links of the graph. None of the links starts and ends in the same node (if there is any we can freely erase it). Links are labelled by representations of the Lorentz  group $(k_e,\rho_{e})$ from the principal series. For a link $e$ we denote by $\ei$ , $\eo$ nodes connected by the link. The link starts at $\eo$ and ends in $\ei$.
\end{itemize}
We will also denote
\begin{equation}
\tau=\frac{1}{2|E_1|}\ .\label{tau-def}
\end{equation}

We restrict ourselves to the case, when all labels are simple tensors. General case can be simply inferred by linearity, because all elements in $\cal F$ are in fact finite linear combinations of such tensors.
Let $\iota_i\in {\cal F}(V^{\kappa_1}_{k_{e_1},\rho_{e_1}}\otimes\cdots\otimes V^{\kappa_n}_{k_{e_n},\rho_{e_n}})$ for $i\in E_0$ be of the form
\[
 \iota_i=v^{i,e_1}\otimes\cdots\otimes v^{i,e_n}.
\]
Vectors $v^{i,e_l}$ (associated with the node $i$ and the link $e_l$ having $i$ as its endpoint) are in the $SU(2)$ subrepresentation of spin $j^{i,e_l}$ of $V^{\kappa_l}_{k_{e_l},\rho_{e_l}}$. The representation $V^{\kappa_l}_{k_{e_l},\rho_{e_l}}$ is equal to $V_{k_{e_l},\rho_{e_l}}$ if the link $e_l$ is outgoing and to the dual $(V_{k_{e_l},\rho_{e_l}})^*$ if the link is incoming to the node $i$.

The evaluation (in our normalization of the measure) is equal to the integral \cite{Engle,Barrett}
\begin{equation}\label{int-tech}
 \int_{SL(2,{\mathbb C})^{|E_0|-1}} \prod_{i\in E_0\setminus\{0\}} \rd g_i  \prod_{e\in E_1} \left( v^{\ei,e},\ T_{k_e,\rho_e}(g_{\ei}^{-1} g_{\eo})\ v^{\eo,e}\right)_{k_e,\rho_e}\ ,
\end{equation}
where $\ei$ and $\eo$ are the nodes connected by the link $e$ and $(\cdot,\cdot)_{k_e,\rho_e}$ is a natural pairing between $(V_{k_{e},\rho_{e}})^*$ and $V_{k_{e},\rho_{e}}$. We put $g_0={\mathcal I}$. 

There exists an antilinear isomorphism
${\cal A}_{k_e,\rho_e}\colon (V_{k_{e},\rho_{e}})^*\xrightarrow{\approx} V_{k_{e},\rho_{e}}$ \cite{Barrett}. It satisfies
\[
 (v^*,v)_{k_e,\rho_e}=\langle {\cal A}_{k_e,\rho_e}(v^*),\ v\rangle_{k_e,\rho_e}\quad \forall{v^*\in (V_{k_{e_l},\rho_{e_l}})^*,v\in V_{k_{e_l},\rho_{e_l}}},
\]
where $\langle\cdot,\cdot\rangle_{k_e,\rho_e}$ is the hermitian inner product.
This map sends the subspace $(V_{k_{e},\rho_{e}})^*(j)$ into $V_{k_{e},\rho_{e}}(j)$ for all $j\geq k_e$.
We can estimate
\begin{equation}\label{est-1}
\begin{split}
|(v^{{\ei},e},\ T_{k_e,\rho_e}(g_{\ei}^{-1} g_{\eo})&\ v^{{\eo},e})_{k_e,\rho_e}|=|\langle {\cal A}_{k_e,\rho_e}(v^{{\ei},e}),\ T_{k_e,\rho_e}(g_{\ei}^{-1} g_{\eo})\ v^{{\eo},e}\rangle_{k_e,\rho_e}|\\
&\leq|{\cal A}_{k_e,\rho_e}(v^{{\ei},e})|_{k_e,\rho_e}|v^{{\eo},e}|_{k_e,\rho_e} \sup |\langle v^{e},\ T_{k_e,\rho_e}(g_{\ei}^{-1} g_{\eo})\ w^{e}\rangle_{k_e,\rho_e}| 
\end{split}
\end{equation}
where the supremum is taken over all $v^{e}\in V_{k_e,\rho_e}(j^{\ei,e})$, and $w^{e}\in V_{k_e,\rho_e}(j^{\eo,e})$ such that $|v^{e}|_{k_e,\rho_e}=|w^{e}|_{k_e,\rho_e}=1$.

The element $g_{\ei}^{-1} g_{\eo}$ can be written as 
\begin{equation}\label{def-r}
 g_{\ei}^{-1} g_{\eo}=u\left[\begin{array}{cc}
                          e^{r/2} & 0\\ 0 & e^{-r/2}
                         \end{array}
\right] v,
\end{equation}
where $u$ and $v$ belong to $SU(2)$. The value of $r>0$ is uniquely determined by this equation. In fact, $r$ can be computed from the equation
\[
 \cosh r=\frac12 \tr g_{\ei}^{-1} g_{\eo}\left(g_{\ei}^{-1} g_{\eo}\right)^{\dagger},
\]
where the trace is in the $SL(2,{\mathbb C})$ defining representation ${\mathbb C}^2$.
Thus, in the estimation \ref{est-1} we have
\[
\sup |\langle v^{e},\ T_{k_e,\rho_e}(g_{\ei}^{-1} g_{\eo})\ w^e\rangle_{k_e,\rho_e}|= \sup \left|\left\langle v^{e},\ T_{k_e,\rho_e}\left(\left[\begin{array}{cc}
                          e^{r/2} & 0\\ 0 & e^{-r/2}
                         \end{array}
\right]\right)\ w^e\right\rangle_{k_e,\rho_e}\right|.
\]
The last term can be estimated by
\[
 \max_{m=-j_e\ldots j_e} \left|\left\langle v^{j^{\ei,e}}_m, T_{k_e,\rho_e}\left(\left[\begin{array}{cc}
                          e^{r/2} & 0\\ 0 & e^{-r/2}
                         \end{array}
\right]\right) v^{j^{\eo,e}}_m\right\rangle_{k_e,\rho_e}\right|
\]
where $v^{j^{\ei,e}}_m$ and $v^{j^{\eo,e}}_m$ are the bases of eigenvectors of $T_{k_e,\rho_e}(L_z)$ (labelled by eigenvalues $m$) in the $SU(2)$ subspaces of $V_{k_e,\rho_e}$ corresponding to $j^{\ei,e}$ and $j^{\eo,e}$. We can restrict ourselves to these matrix elements because $T_{k_e,\rho_e}\left(\left[\begin{array}{cc}
                          e^{r/2} & 0\\ 0 & e^{-r/2}
                         \end{array}
\right]\right)$ is diagonal in these bases. We use notation $j_e=\min j^{\ei,e},j^{\eo,e}$.

According to \cite{Ruhl} (Eq. 3-32) we have
\begin{equation}\label{Ruhl-int}
 \begin{split}
&\left\langle 
v^{j^{\ei,e}}_m, T_{k_e,\rho_e}\left(
\left[\begin{array}{cc}
                          e^{r/2} & 0\\ 0 & e^{-r/2}
                         \end{array}\right]\right)
v^{j^{\eo,e}}_m\right\rangle_{k_e,\rho_e}
=\sqrt{(2j_{\ei}+1)(2j_{\eo}+1)}\\
&\phantom{uuuu}\int_0^1 \rd t\ d^{j^{\ei,e}}_{k_em}(2t-1) d^{j^{\eo,e}}_{k_em}\left(\frac{te^{-r}-(1-t)e^r}{te^{-r}+(1-t)e^r}\right)(te^{-r}+(1-t)e^r)^{i\rho_e/2}\frac{1}{te^{-r}+(1-t)e^r}\ ,
 \end{split}
\end{equation}
where $d^{j}_{k_em}$ is equal to the matrix element in the representation of $SU(2)$ with spin $j$
\[
 d^{j}_{k_em}(\cos\theta)=\left\langle v^j_{k_e}, T_j\left[\begin{array}{cc}
                                                      \cos\frac{\theta}{2} & \sin\frac{\theta}{2}\\
-\sin\frac{\theta}{2}& \cos\frac{\theta}{2}
                                                     \end{array}
\right]v^j_m\right\rangle\ .
\]
We can estimate integral \ref{Ruhl-int} by
\begin{equation}\label{est-D}
 C_m\int_0^1 \rd t \frac{1}{te^{-r}+(1-t)e^r}=C_m\frac{r}{\sinh r}\leq D\frac{1}{(\cosh r)^{1-\tau}}\ ,
\end{equation}
where $C_m$ ($m=-j_e,\ldots, j_e$) and $D$ ($m$ independent) are constants. We should stress that $D$ depends on the choice of $j^{\ei,e}$, $j^{\eo,e}$ and on $|E_1|$ but can be chosen $\rho_e$ independent.

The measure on $SL(2,{\mathbb C})$ can be decomposed into measure $\rd u$ on $SU(2)$ and $\rd \mu$ on the space of boosts (every element $g\in SL(2,{\mathbb C})$ can be uniquely written as $g=\mu u$ where $u\in SU(2)$ and $\mu$ is a boost)
\begin{equation}\label{dg1}
 \rd g=\rd u\rd \mu,\ \ \int_{SU(2)}\rd u=1\ .
\end{equation}
Elements of the space of boosts can be identified with either the coset space $SL(2,{\mathbb C})/SU(2)$ or with $3$ dimensional unit hyperboloid ${\cal H}_+$ (future part). Element of the hyperboloid is associated with a unique boost that transform vector $(1,0,0,0)$ into it.

On the hyperboloid we have natural coordinates given by $\eta=(\eta_0,\vec{\eta})$ with a constraint $\eta_0^2-\vec{\eta}^2=1$. The measure is as follows
\begin{equation}\label{dg2}
 \rd\mu
=2\delta(\eta_0^2-\vec{\eta}^2-1)\rd\eta_0\cdots\rd\eta_3
=\frac{1}{\eta_0}\rd\eta_1\rd\eta_2\rd\eta_3\ .
\end{equation}
The distance $r$ between two points with coordinates $\eta^0,\eta^1$ can be computed from the equation
\begin{equation}
 \cosh r=\eta^0_0\eta^1_0-\vec\eta^0\vec\eta^1\ .
\end{equation}
Alternatively, one can take two elements of $g^0,g^1\in SL(2,{\mathbb C})$ whose classes in the coset space correspond to points $\eta^0$, $\eta^1$, then $r$ is defined by \ref{def-r} for $(g^0)^{-1}g^1$.

Now we can estimate (up to a constant factor depending on $|E_1|$ and on the labels of nodes) the integral \ref{int-tech} by
\[
\int\prod_{i\in E_0\setminus\{0\}} \rd u_i\rd\mu_i \prod_{e\in E_1} \left(\frac{1}{\cosh r_{\ei\eo}}\right)^{1-\tau}=\int\prod_{i\in E_0\setminus\{0\}} \rd\mu_i \prod_{e\in E_1} \left(\frac{1}{\cosh r_{\ei\eo}}\right)^{1-\tau}\ ,
\]
where $r_{\ei\eo}$ is the distance between nodes $\ei$ and $\eo$ on the unit hyperboloid. The integrated function is $u_i$ independent and we could integrate out $SU(2)$ factors.

Hence, it is enough to prove
\begin{thm}\label{tech}
For every 3-edge-connected graph the integral
\begin{equation}
\int_{{\cal H}_+^{|E_0|-1}}\prod_{i\in E_0\setminus\{0\}} \rd\mu_i \prod_{e\in E_1} \left(\frac{1}{\cosh r_{\ei\eo}}\right)^{1-\tau},
\end{equation}
where $r_{\ei\eo}$ is the distance on the hyperboloid between nodes $\ei$ and $\eo$,
is finite.
\end{thm}
The proof of this theorem occupies Sec. \ref{proof}.

\section{Proof of the theorem \ref{tech}}\label{proof}
In this section we will prove the technical result, theorem \ref{tech}, that is main ingredient of the proof of Baez-Barrett conjecture.
\vskip 0.3cm

\noindent {\it Sketch of the proof:} The main idea is to estimate the integral from Theorem \ref{tech} by a sequence of other integrals introduced in \ref{sec-integrals} (see Lemma~\ref{lm-est1}). It is done by a division of the range of integration into some smaller regions. The integral in each region is estimated separately. In order to do that, we need to introduce some methods of estimations (see Sec.~\ref{sec-prelim} and \ref{sec-estimates}). Finally, by an inductive procedure we estimate the initial integral by a finite constant (see Lemma~\ref{cond-d}). Conditions for Lemma~\ref{lm-est1} to hold (see also Lemma~\ref{cond-d}) are verified in Sec.~\ref{sec-graph}. Here the 3-edge-connectedness comes into play (see Lemma~\ref{lemma-aux3}).

\subsection{Preliminaries}\label{sec-prelim}
We will use the following abbreviations for $f,g$ nonnegative functions or measures
\begin{itemize}
 \item $f\approx g$ (equivalence)
if there exists constant $C>0$ such that $\frac{1}{C}g\leq f\leq Cg$ for whole range of variables.
\item $f\preceq g$ (estimate)
if there exists constant $C>0$ such that $f\leq Cg$ for whole range of variables.
\end{itemize}
Suppose that we have positive functions $f_1\approx f_2$ and $g_1\approx g_2$ then
\[
 f_1g_1\approx f_2g_2,\ f_1^\alpha\approx f_2^\alpha,\ {\rm for\ } \alpha\in{\mathbb R}\ .
\]
Similarly, for $f_1\preceq f_2$ and $g_1\preceq g_2$
\[
 f_1g_1\preceq f_2g_2,\ f_1^\alpha\preceq f_2^\alpha,\ {\rm for\ } \alpha>0\ .
\]
{\it Remark:} Later on we will use the following convention: By $f(\{y_l\}_{l=0\ldots n})$ we mean a function depending on the set $\{y_l\}_{l=0\ldots n}$ of variables. If we write (for example)
\[
g(z)=f(\{y_l=0\}_{l\not=3},y_3=z),
\]
then function $g(z)$ is obtained by putting in $f$ all $y_l=0$ except $y_3=z$.

We introduce new coordinates $\epsilon$, $\vec{\xi}$ on the hyperboloid
\begin{itemize}
 \item $\epsilon=\frac{1}{\eta_0}$ has range $(0,1]$
\item $\vec{\xi}=\frac{\vec{\eta}}{|\eta|}$ 
are coordinates on the two dimensional sphere ($3$-dimensional vector of norm $1$), $|\eta|=\sqrt{\eta_1^2+\eta_2^2+\eta_3^2}$.
\end{itemize}
We have the following expression for the measure
\[
 \rd\mu=\frac{\sqrt{1-\epsilon^2}}{\epsilon^3} \rd\epsilon\rd^2\xi\ ,
\]
where $\rd^2\xi$ is the normal measure on the sphere.

Finally, the measure is estimated by
\begin{equation}
\rd\mu\preceq \epsilon^{-3}\ \rd\epsilon \rd^2\xi 
\end{equation}

\begin{lm}
We have equivalence for the distance $r_{ij}$ on the hyperboloid between two nodes $i$ and $j$ (described by $(\epsilon_i,\vec{\xi}_i)$ and $(\epsilon_j,\vec{\xi}_j)$)
\begin{equation}
 \frac{1}{\cosh r_{ij}}\approx \frac{\epsilon_i\epsilon_j}{\epsilon_i^2+\epsilon_j^2+\theta_{ij}^2},
\end{equation}
where $\theta_{ij}$ is the distance on the sphere between points given by $\vec{\xi}_i,\vec{\xi}_j$.
\end{lm}

\begin{proof}
We compute
\[
\frac{1}{\cosh r_{ij}}= \frac{\epsilon_i\epsilon_j} {1-\sqrt{1-\epsilon_i^2}\sqrt{1-\epsilon_j^2}(\vec{\xi}_i\vec{\xi}_j)}=\frac{\epsilon_i\epsilon_j}{1-\sqrt{1-\epsilon_i^2}\sqrt{1-\epsilon_j^2}\cos\theta_{ij}},
\]
where $r_{ij}$ is the hyperbolic distance between points described by $(\epsilon_i,\vec{\xi}_i)$ and $(\epsilon_j,\vec{\xi}_j)$. The scalar product between two vectors $\vec{\xi}_i$, $\vec{\xi}_j$ is denoted by $\vec{\xi}_i\vec{\xi}_j=\cos\theta_{ij}$.
The ratio
\[
 \frac{\frac{\epsilon_i\epsilon_j}{\epsilon_i^2+\epsilon_j^2+\theta_{ij}^2}}
{\frac{\epsilon_i\epsilon_j}{1-\sqrt{1-\epsilon_i^2}\sqrt{1-\epsilon_j^2}\cos\theta_{ij}}}
=\frac{1-\sqrt{1-\epsilon_i^2}\sqrt{1-\epsilon_j^2}\cos\theta_{ij}}{\epsilon_i^2+\epsilon_j^2+\theta_{ij}^2}
\]
can blow up (or go to zero) only for $\epsilon_i=\epsilon_j=\theta_{ij}=0$. Due to the compact range of parameters (we extend $\epsilon\in[0,1]$), in order to prove the lemma it is enough to show that there exists a limit
\[
\lim_{(\epsilon_i,\epsilon_j,\theta_{ij})\rightarrow (0,0,0)} \frac{1-\sqrt{1-\epsilon_i^2}\sqrt{1-\epsilon_j^2}\cos\theta_{ij}}{\epsilon_i^2+\epsilon_j^2+\theta_{ij}^2}\ .
\]
Let us introduce a distance $r$ from the point $(0,0,0)$, $r=\sqrt{\epsilon_i^2+\epsilon_j^2+\theta_{ij}^2}$. We introduce
\[
\epsilon_i=r\tilde\epsilon_i,\ \ \epsilon_j=r\tilde\epsilon_j,\ \ \theta_{ij}=r\tilde\theta_{ij},\ \ 
\tilde\epsilon_i^2+\tilde\epsilon_j^2+\tilde\theta_{ij}^2=1.
\]
In the limit $r\rightarrow 0$ we have
\[
 \lim_{r\rightarrow 0}
\frac{1-\sqrt{1-r^2\tilde\epsilon_i^2}\sqrt{1-r^2\tilde\epsilon_j^2}\cos(r^2\tilde\theta_{ij})}{r^2(\tilde\epsilon_i^2+\tilde\epsilon_j^2+\tilde\theta_{ij}^2)}=
\lim_{r\rightarrow 0}
\frac{\frac{1}{2}r^2(\tilde\epsilon_i^2+\tilde\epsilon_j^2+\tilde\theta_{ij}^2)+O(r^4)}{r^2(\tilde\epsilon_0^2+\tilde\epsilon_1^2+\tilde\theta_{01}^2)}=\frac12\ .
\]
\end{proof}
As a result we obtain
\[
 \left(\frac{1}{\cosh r_{ij}}\right)^{1-\tau}\approx \left(\frac{\epsilon_i\epsilon_j}{\epsilon_i^2+\epsilon_j^2+\theta_{ij}^2}\right)^{1-\tau}\ .
\]
The integral, that we would like to estimate, is estimated ($\preceq$) by an integral
\begin{equation}
 \label{int-simpl}
 \int \prod_{i\not=0} \frac{\rd\epsilon_i}{\epsilon_i^3} \prod_{i} \rd^2\xi_i \prod_{e\in E_1} 
\left(\frac{\epsilon_{\ei}\epsilon_{\eo}}{\epsilon_{\ei}^2+\epsilon_{\eo}^2+\theta_{\ei\eo}^2}\right)^{1-\tau},
\end{equation}
where $\theta_{\ei\eo}$ is the distance on the sphere between the end points $\vec{\xi}_{\ei}$ and $\vec{\xi}_{\eo}$ of the link $e$.

We need to fix $\epsilon_0=1$ and similarly we need to fix $\vec{\xi}_0$. Integration over the latter does not change the value because the integral \ref{int-simpl} is invariant with respect to the simultaneous rotation of all $\vec{\xi}$. So, the integration over $\vec{\xi}_0$ gives only the common factor of the area of the sphere (finite), and we will perform it.

\subsection{Integrals}\label{sec-integrals}

We consider a set of sequences of pairwise different elements from the disjoint union of $E_0\setminus\{0\}$ (the set of nodes without node $0$) and the set of links $E_1$
\[
 \left(E_0\setminus\{0\}\right)\amalg E_1,
\]
such that if we consider the set of links belonging to the sequence then there is no loops i.e. disregarding nodes, links form a tree (maybe disconnected).
The set of such sequences we denote by $\ciag$. The length of a sequence $I\in\ciag$ is denoted by $|I|$. Elements of the sequence $I$ that belong to $E_0$ we will call nodes. Elements belonging to $E_1$ we call links. We denote by $I+ \{i\}$ (respectively $I+\{e\}$) the sequence $I$ with added node $i$ (link $e$, respectively) at the end.

With every element $I\in\ciag$ we associate the graph ${\cal G}_I$ constructed as follows: We take the initial graph and merge every two nodes that are connected by a path of links from $I$. In the graph just obtained we erase every link that starts and ends in the same node. Notice, that graph ${\cal G}_\emptyset$ is in fact the initial graph. We will denote by $E_0({\cal G}_I)$ (and $E_1({\cal G}_I)$) the set of nodes (and respectively the set of links) of the graph ${\cal G}_I$.

We can regard each node $\ti\in E_0({\cal G}_I)$ as a set of nodes of the initial graph that are merged into $\ti$. Let sequence $J$ be a prolongation of the sequence $I$. We will write $\ti\subset\ti'$ for $\ti\in E_0({\cal G}_I)$ and $\ti'\in E_0({\cal G}_J)$ if the relation of inclusion holds for the corresponding sets. Moreover, for every $\ti\in E_0({\cal G}_I)$ there is a unique node $[\ti]$ in ${\cal G}_J$ such that $\ti\subset [\ti]$. For each $\ti\in E_0({\cal G}_I)$ the number $N_{\ti}$ is the number of nodes of the initial graph that are merged into $\ti$.

We will also write $(\ti,\tj) \in E_1({\cal G}_I)$ if $\ti\in E_0({\cal G}_I)$ and $\tj\in E_0({\cal G}_I)$ are connected by a link in ${\cal G}_I$. For any link $f\in E_1({\cal G}_I)$ we denote by $\tfi$ and $\tfo$ nodes of ${\cal G}_I$ connected by $f$. 
\vskip 0.35cm

We associate with $I$ a positive function $f_I(\chi,\{\epsilon_i\}_{i\notin I},\{\theta_{\ti\tj}\}_{\ti,\tj\in E_0({\cal G}_I),\: (\ti,\tj)\in E_1({\cal G}_I)})$. Notice that $\epsilon_i$ are labelled by nodes of the initial graph, but $\theta_{\ti\tj}$ by unordered pairs of nodes of the graph ${\cal G}_I$. We associate with $I$ also an integral
\begin{equation}\label{def-integrals}
 \I_I=\int \prod_{i\notin I} \rd\epsilon_i \prod_{\ti\in E_0({\cal G}_I)} \rd^2\xi_{\ti}\: f_I(\chi=\min\left[\{\epsilon_i\}_{i\notin I}\cup\{\theta_{\tei\teo}\}_{e\in E_1({\cal G}_I)}\right],\{\epsilon_i\}_{i\notin I},\{\theta_{\ti\tj}\}_{(\ti,\tj)\in E_1({\cal G}_I)})
\end{equation}
In the integral $0\leq \theta_{\ti\tj}\leq \pi$ is the distance on the sphere between points $\vec{\xi}_\ti$ and $\vec{\xi}_\tj$ and can be computed by the equation
$\cos \theta_{\ti\tj}=\vec{\xi}_\ti\vec{\xi}_\tj$. 

The function is defined inductively. We take as $f_{\emptyset}$ the function
\[
 \prod_{i\in E_0}\frac{1}{\epsilon_i^3}\prod_{e\in E_1} 
\left(\frac{\epsilon_{\ei}\epsilon_{\eo}}{\epsilon_{\ei}^2+\epsilon_{\eo}^2+\theta_{\ei\eo}^2}\right)^{1-\tau}
\]
This function does not depend on $\chi$. Because the graph ${\cal G}_\emptyset$ is the initial graph, we use $\ei,\eo$ instead of $\tei,\teo$ in the indices of $\theta_{\ei\eo}$.

Suppose we have constructed $f_I$ for the given $I$ then 
\begin{itemize}
 \item for $J=I+\{k\}$, $k$ being a node, we put
\[
 f_J(\chi,\{\epsilon_i\}_{i\notin J},\{\theta_{\ti\tj}\}_{(\ti,\tj)\in E_1({\cal G}_I)})=\chi f_I(\chi,\{\epsilon_i\}_{i\notin I\cup\{k\}},\epsilon_k=\chi,\{\theta_{\ti\tj}\}_{(\ti,\tj)\in E_1({\cal G}_I)})\ ,
\]
\item for $J=I+\{e\}$, with $e\in E_1$,
\[
 f_J(\chi,\{\epsilon_i\}_{i\notin I},\{\theta_{\ti'\tj'}\}_{(\ti',\tj')\in E_1({\cal G}_J)})=
\chi^2f_I(\chi,\{\epsilon_i\}_{i\notin I},
\{\theta_{\ti\tj}=\theta_{[\ti][\tj]}\}_{([\ti],[\tj])\in E_1({\cal G}_J)},\theta_{\tei\teo}=\chi)
\]
Notice, that the graph ${\cal G}_J$ is obtained from ${\cal G}_I$ by merging nodes $\tei$ and $\teo$ (and erasing suitable links), so in the right-hand side function above we put as $\theta_{\ti\tj}$
\[
 \theta_{\ti\tj}=\left\{\begin{array}{ll}
                         \chi & {\rm if}\ \{\ti,\tj\}=\{\tei,\teo\}\\
                         \theta_{[\ti][\tj]} & {\rm otherwise.}
                        \end{array}\right.
\]
\end{itemize}
The integral \ref{int-simpl} corresponds to the empty sequence and is equal to $\I_\emptyset$.

\subsection{Estimates}\label{sec-estimates}

Let us define a class of functions $\Psi$.
\begin{df}
Function $f(\vec{y})$ of the variables $y_0,\ldots,y_n\ >0$ belongs to the class $\Psi$ if it is a product of some powers (real) of $y_\alpha$, $\alpha=0,\ldots n$ and (real positive) powers of  $\frac{1}{y_\alpha^2+y_\beta^2+y_\gamma^2}$, $\alpha,\beta,\gamma=0\ldots n$.
\end{df}

All the functions $f_I$ for $I\in\ciag$ belong to the class $\Psi$. This class has some useful properties.

\begin{lm}\label{lm-lambda}
If $f(\vec{y})\in\Psi$ then there exists a unique $\lambda\in{\mathbb R}$ such that for every fixed values of $y_1,y_2,\ldots,y_n>0$ there exist a nonzero, finite limit
\[
 \lim_{y_0\rightarrow 0} \frac{f(\vec{y})}{y_0^\lambda}.
\]
\end{lm}

\begin{proof}
It is enough to show that for every factor in $f(\vec{y})$ there exists such a unique constant. For $y_\alpha^{\kappa}$ this constant is $\kappa$ (if $\alpha=0$) or $0$ (if $\alpha\not=0$). For
$\left(\frac{1}{y_\alpha^2+y_\beta^2+y_\gamma^2}\right)^\kappa$ it is $-2\kappa$ (if $\alpha=\beta=\gamma=0$) or $0$ (otherwise).
\end{proof}

We will use it in the definition
\begin{df}
For $f(\vec{y})\in\Psi$ the value $\lambda$ obtained in the lemma \ref{lm-lambda} will be denoted by 
\[
 \Val_{y_0} f:=\lambda\ .
\]
\end{df}

\begin{lm}\label{integral}
Let $f(\vec{y})\in\Psi$ and $C>0$ then for every $\vec{y}$ such that $y_0\leq\min\{y_1,\dots,y_n\}$ we have
\begin{equation}
 \int_0^{Cy_0} f(\vec{z}) \rd z_0\approx f(\vec{y})y_0\ ,
\end{equation}
if $\Val_{y_0}(f(\vec{y})y_0)>0$. In the integration remaining variables $y_1,\ldots, y_n$ are fixed. We treat the resulting integral as a function of parameters $\vec{y}$ in the domain where the condition $y_0\leq\min\{y_1,\dots,y_n\}$ is satisfied (and in this sense we use $\approx$).
\end{lm}

\begin{proof}
The function $f$ can be written as
\[
 f(\vec{z})=F(\vec{z}) z_0^\lambda\ ,
\]
where $F(\vec{z})$ has nonzero limit for $z_0\rightarrow 0$ (all other parameters fixed) and $\lambda=\Val_{y_0}f$.
We have for $0<z_0\leq Cy_0$
\begin{equation}\label{eq-zy}
 \frac{1}{z_\alpha^2+z_\beta^2+z_\gamma^2}\approx \frac{1}{y_\alpha^2+y_\beta^2+y_\gamma^2}
\end{equation}
on the domain of integration, if at least one of $\alpha$, $\beta$ or $\gamma$ is not $0$. 
Indeed, then
\[
 \frac{1}{1+C^{2}}\frac{1}{y_\alpha^2+y_\beta^2+y_\gamma^2}\leq \frac{1}{z_\alpha^2+z_\beta^2+z_\gamma^2}\leq 3(1+C^2) \frac{1}{y_\alpha^2+y_\beta^2+y_\gamma^2}\ .
\]
If all $\alpha,\beta,\gamma$ in the expression $\frac{1}{z_\alpha^2+z_\beta^2+z_\gamma^2}$ are $0$ then it will not appear in $F(\vec{z})$ (because it is a power of $z_0$). Similarly, $z_\alpha\approx y_\alpha$ for $\alpha\not=0$ (because then $z_\alpha=y_\alpha$). The integral is thus equivalent to
\[
 \int_0^{Cy_0} f(\vec{z}) \rd z_0\approx \int_0^{Cy_0} F(\vec{y})z_0^\lambda\rd z_0\approx F(\vec{y})y_0^{\lambda+1}=f(\vec{y})y_0
\]
if $\Val_{y_0} (f(\vec{y})y_0)=\lambda+1>0$.
\end{proof}

\begin{lm}\label{lemma-aprox2}
Suppose that we have $f(\vec{y})\in\Psi$ on some domain and the family of functions ${z_\alpha}(\vec{y})_{\alpha=0\ldots n}$ satisfying
\[
 \forall_\alpha {z_\alpha}(\vec{y})\approx y_\alpha\ ,
\]
then $f(\vec{y})\approx f(\vec{z}(\vec{y}))$ on the domain.
\end{lm}

\begin{proof}
Let the constant of equivalence $y_\alpha\approx z_\alpha$ be denoted by $C_\alpha$, then we have
\[
 \frac{1}{C_\alpha^2+C_\beta^2+C_\gamma^2}\frac{1}{z_\alpha^2+z_\beta^2+z_\gamma^2}\leq
\frac{1}{y_\alpha^2+y_\beta^2+y_\gamma^2}\leq \frac{{C_\alpha^2+C_\beta^2+C_\gamma^2}}{z_\alpha^2+z_\beta^2+z_\gamma^2}\ .
\]
Powers of $y_\alpha$ can be treated similarly. Because relation $\approx$ is preserved by multiplication the lemma follows.
\end{proof}

\subsection{Inequalities}

For the given $I\in\ciag$ we denote by $d(I)$
\[
 d(I):=\Val_{\chi} f_I\ .
\]
We have the following lemma

\begin{lm}\label{lm-est1}
For a given sequence $I\in\ciag$
\begin{equation}
 \I_I\preceq 1+\sum_{|J|=|I|+1} \I_J\ .
\end{equation}
if $d(J)>0$ for all $J\in\ciag$ such that $|J|=|I|+1$.
\end{lm}

\noindent{\it Remark:} In fact, summation can be restricted to the one step prolongations of $I$. The integrals $\I_J$ are positive numbers, maybe infinite. As usual, by $\preceq$ we mean inequality up to a constant finite factor.

\begin{proof}
We divide the area of integration in $\I_I$ into regions (maybe overlaping), depending on what is the smallest parameter from the set
\begin{equation}
 \{ \epsilon_i\colon i\notin I\}\cup \{\theta_{\ti\tj}\colon \ti,\tj\in E_0({\cal G}_I),\:  (\ti,\tj)\in E_1({\cal G}_I)\}
\end{equation}
In each region we estimate the integral by other methods.

There are three different sorts of cases
\begin{itemize}
 \item the smallest parameter is $\epsilon_k$ for some $k\not=0$,
 \item the smallest parameter is $\theta_{\tei\teo}$ for some ${e}\in E_1({\cal G}_I)$.
 \item the smallest parameter is $\epsilon_0$,
\end{itemize}
In the first case we can write the integration over the region where $\epsilon_k\leq \chi_{I+\{k\}}=\min \{ \epsilon_i\colon i\notin I,\ i\not=k\}\cup \{\theta_{\ti\tj}\colon \ti,\tj\in E_0({\cal G}_I),\: (\ti,\tj)\in E_1({\cal G}_I)\}$ as
\[ 
 \int \prod_{i\notin I,\ i\not=k} \rd\epsilon_i \prod_{\ti\in E_0({\cal G}_I)} \rd^2\xi_{\ti} \int_0^{\chi_{I+\{k\}}}\rd \epsilon_k\: f_I(\chi=\epsilon_k,\{\epsilon_i\}_{i\notin I\cup \{k\}},\epsilon_k,\{\theta_{\ti\tj}\}_{(\ti,\tj)\in E_1({\cal G}_I)})
\]
because in this region $\chi=\epsilon_k$. We have for $\chi'\leq\chi_{I+\{k\}}$
\[
\begin{split}
 \int_0^{\chi'}\rd\epsilon_k\: f_I(\chi=\epsilon_k,&\{\epsilon_i\}_{i\notin I\cup\{k\}},\epsilon_k,\{\theta_{\ti\tj}\}_{(\ti,\tj)\in E_1({\cal G}_I)})\approx\\
&\approx\chi' f_I(\chi=\chi',\{\epsilon_i\}_{i\notin I\cup\{k\}},\epsilon_k=\chi',\{\theta_{\ti\tj}\}_{(\ti,\tj)\in E_1({\cal G}_I)})
\end{split}
\]
by lemma \ref{integral} if $d(I+\{k\})>0$. The resulting integral is equivalent as a function to $f_{I+\{k\}}$. The remaining integrations are as in $\I_{I+\{k\}}$ (notice that ${\cal G}_I={\cal G}_{I+\{k\}}$)
\[
\int \prod_{i\notin I\cup\{k\}} \rd\epsilon_i \prod_{\ti\in E_0({\cal G}_I)} \rd^2\xi_{\ti}\ 
\chi_{I+\{k\}}
f_I(\chi=\chi_{I+\{k\}},\{\epsilon_i\}_{i\notin I\cup\{k\}},\epsilon_k=\chi_{I+\{k\}},\{\theta_{\ti\tj}\}_{(\ti,\tj)\in E_1({\cal G}_I)})
\]
and so the integral on this region can be estimated ($\preceq$) by $\I_{I+\{k\}}$ (in the formula above $\chi_{I+\{k\}}$ is a function of the remaining variables).

In the second case, notice that $e$ connects two nodes $\tei$ and $\teo$ that are distinct in ${\cal G}_I$ (we erased all links starting and ending in the same node). We can regard $e$ also as a link in the initial graph. Notice then that $I+\{e\}\in\ciag$ because there are no loops of links in $I+\{e\}$.
We introduce new variables: 
\begin{itemize}
\item $\vec{\xi}_{e_e}$ is the coordinates of the center of the link $e$, that will correspond to the coordinates of the merged node. Explicitly, $\vec{\xi}_{e_e}=\frac{\vec{\xi}_{\tei}+\vec{\xi}_{\teo}}{|\vec{\xi}_{\tei}+\vec{\xi}_{\teo}|}$ regarded as a $3$-dimensional vector normalized to $1$. It is well defined because $\theta_{\tei\teo}\leq \epsilon_0=1$, so $\vec{\xi}_{\tei}+\vec{\xi}_{\teo}\not=0$.
\item $\phi$ is the angle between the link and some chosen axis hitched in $\vec{\xi}_{e_e}$. In every point on the sphere we choose independently this axis. Notice that $\theta_e\rd\phi$ is continuous $1$-form even if $\phi$ is not continuous function.
\end{itemize}
We perform a change of variables (by $e_e$ we denote merged node $[\tei]=[\teo]$ and $\theta_e=\theta_{\tei\teo}$)
\[
\begin{split}
 \{\epsilon_i\}_{i\notin I},\ \{\vec{\xi}_{\ti}\}_{\ti\in E_0({\cal G}_I),\: \ti\not\in\{\tei,\teo\}},&\ 
\vec{\xi}_{\tei},\vec{\xi}_{\teo}\rightarrow\\
 &\rightarrow\{\epsilon_i\}_{i\notin I},\ \{\vec{\xi}_{\ti'}\}_{\ti'\in E_0({\cal G}_{I+\{e\}}),\:  \ti'\not=e_e},\ \vec{\xi}_{e_e},\theta_e,\phi\ ,
\end{split}
\]
where we change, in fact, only $\vec{\xi}_{\tei},\vec{\xi}_{\teo}\rightarrow\vec{\xi}_{e_e},\theta_e,\phi$ and put $\vec{\xi}_{[\ti]}=\vec{\xi}_{\ti}$ for $\ti\not=\tei,\ \teo$.
The measure changes as follows
\[
 \rd^2\xi_{\tei}\rd^2\xi_{\teo}\preceq \theta_e\rd \theta_e\rd\phi\rd^2\xi_{e_e}
\]
In the new variables the integral on the region is estimated by
\[
\int \prod_{i\notin I} \rd\epsilon_i \prod_{\ti'\in E_0({\cal G}_{I+\{e\}})} \rd^2\xi_{\ti'}\: \int_0^{2\pi}\rd\phi\int_0^{\tilde{\chi}_{I+\{e\}}}\theta_e\rd\theta_e\  f_I(\chi=\theta_e,\{\epsilon_i\}_{i\notin I},\{\theta_{\ti\tj}\}_{(\ti,\tj)\in E_1({\cal G}_I)})\ ,
\]
where $\tilde{\chi}_{I+\{e\}}=\min\{\epsilon_i\colon i\notin I\}\cup\{\theta_{\tfi\tfo}\colon f\in E_1({\cal G}_I),\ \{\tfi,\tfo\}\not=\{\tei,\teo\}\}$.
For each node $\tk\in E_1({\cal G}_I)$ different from $\tei,\teo$  we have in this region
\[
 \theta_{\tk\tei}\approx\theta_{\tk e_e}\qquad \theta_{\tk\teo}\approx\theta_{\tk e_e}\ ,
\]
where $\theta_{\tk e_e}$ is the distance on the sphere to new merged node $e_e$.
Indeed, we have
\[
 \theta_{\tk e_e}\leq \theta_{\tk\tei}+\theta_{\tei e_e}=\theta_{\tk\tei}+\frac{1}{2}\theta_e\leq \frac{3}{2}\theta_{\tk\tei},\qquad
\frac{1}{2}\theta_{\tk\tei}\leq\theta_{\tk\tei}-\frac{1}{2}\theta_e\leq\theta_{\tk e_e}
\]
and similarly for $\theta_{\tk\teo}$ because $\theta_e$ is the smallest parameter among $\theta_{\tfi\tfo}$, $f\in E_1({\cal G}_I)$. A constant of equivalence ($\approx$) is equal to $2$. We obtain from lemma \ref{lemma-aprox2} that
\[
f_I(\chi,\{\epsilon_i\}_{i\notin I},\{\theta_{\ti\tj}\}_{(\ti,\tj)\in E_1({\cal G}_I)})\approx
 f_I(\chi,\{\epsilon_i\}_{i\notin I},\{\theta_{\ti\tj}=\theta_{[\ti][\tj]}\}_{([\ti],[\tj])\in E_1({\cal G}_{I+\{e\}})},\theta_{\tei\teo})
\]
In the righthand side function above we put $\theta_{[\ti][\tj]}$ in place of $\theta_{\ti\tj}$ if $\ti,\tj$ are not merged ($[\ti]\not=[\tj]$). Otherwise ($\{\ti,\tj\}=\{\tei,\teo\}$) we put $\theta_{\tei\teo}$. We can estimate ($\preceq$) the integral over the region by
\[
\begin{split}
\int \prod_{i\notin I} \rd\epsilon_i \prod_{\ti'\in E_0({\cal G}_{I+\{e\}})} \rd^2\xi_{\ti'}\: &\int_0^{2\pi}\rd\phi\int_0^{\tilde{\chi}_{I+\{e\}}}\theta_e\rd\theta_e\\
&f_I(\chi=\theta_e,\{\epsilon_i\}_{i\notin I},\{\theta_{\ti\tj}=\theta_{[\ti][\tj]}\}_{([\ti],[\tj])\in E_1({\cal G}_{I+\{e\}})},\theta_{\tei\teo}=\theta_{e}) 
\end{split}
\]
We integrate this function first over variables $\theta_e$ and $\phi$. Integration is performed over the region where $\theta_e$ is the smallest parameter ( $\theta_e\leq\min\{\epsilon_i\colon i\notin I\}\cup\{\theta_{\tfi\tfo}\colon f\in E_1({\cal G}_I),\ \{\tfi,\tfo\}\not=\{\tei,\teo\}\}$). In the new variables $\{\epsilon_k\}_{k\notin I}\cup\{\vec{\xi}_{[\ti]}\}_{[\ti]\in E_0({\cal G}_{I+\{e\}})}$ it is only condition on the range of $\theta_e$. However, it is difficult to express the function $\tilde{\chi}_{I+\{e\}}$ in terms of these variables. Nevertheless, we have 
$$
\tilde{\chi}_{I+\{e\}}\leq {2} \min\{\epsilon_i\colon i\notin I\}\cup\{\theta_{\ti'\tj'}\colon (\ti',\tj')\in E_1({\cal G}_{I+\{e\}})\}=2\chi_{I+\{e\}}\ ,
$$
where we introduced $\chi_{I+\{e\}}=\min\{\epsilon_i\colon i\notin I\}\cup\{\theta_{\ti'\tj'}\colon (\ti',\tj')\in E_1({\cal G}_{I+\{e\}})\}$.
If we spread the integration of $\theta_e$ over the region up to $2 \chi_{I+\{e\}}$ then the resulting integral will grow (we integrate a positive function). Thus, we can estimate ($\preceq$) the previous integral by
\[
\begin{split}
\int \prod_{i\notin I} \rd\epsilon_i \prod_{\ti'\in E_0({\cal G}_{I+\{e\}})} \rd^2\xi_{\ti'}\: &\int_0^{2\pi}\rd\phi\int_0^{2{\chi}_{I+\{e\}}}\theta_e\rd\theta_e\\
&f_I(\chi=\theta_e,\{\epsilon_i\}_{i\notin I},\{\theta_{\ti\tj}=\theta_{[\ti][\tj]}\}_{([\ti],[\tj])\in E_1({\cal G}_{I+\{e\}})},\theta_{\tei\teo}=\theta_{e})
\end{split}
\]
We perform the integration over $\phi$ and $\theta_e$ then for $\chi'\leq \chi_{I+\{e\}}$
\[
\begin{split}
 \underbrace{\int_0^{2\pi}\rd\phi}_{2\pi}\int_0^{2\chi'}\theta_e\rd\theta_e\  f_I(\chi=\theta_e,&\{\epsilon_i\}_{i\notin I},\{\theta_{\ti\tj}=\theta_{[\ti][\tj]}\}_{([\ti],[\tj])\in E_1({\cal G}_{I+\{e\}})}, \theta_{\tei\teo}=\theta_e)\approx\\
&\approx(\chi')^2f_I(\chi=\chi',\{\epsilon_i\}_{i\notin I},\{\theta_{\ti\tj}=\theta_{[\ti][\tj]}\}_{([\ti],[\tj])\in E_1({\cal G}_{I+\{e\}})}, \theta_{\tei\teo}=\chi').
\end{split}
\]
As before we put $\theta_{\ti\tj}=\theta_{[\ti][\tj]}$ if $\{\ti,\tj\}\not=\{\tei,\teo\}$.
One can perform first integration because the function is $\phi$ independent. We applied lemma \ref{integral} under condition that $d(I+\{e\})>0$. The resulting function is equivalent to $f_{I+\{e\}}$. The remaining integrations are as in $\I_{I+\{e\}}$ ($\vec{\xi}_{e_e}$ corresponds to the new merged node $[\tei]=[\teo]$ of ${\cal G}_{I+\{e\}}$)
\[
\begin{split}
\int \prod_{i\notin I} \rd\epsilon_i \prod_{\ti'\in E_0({\cal G}_{I+\{e\}})} \rd^2\xi_{\ti'}\ 
\chi_{I+\{e\}}^2f_I(\chi&=\chi_{I+\{e\}},\{\epsilon_i\}_{i\notin I},\\
&\{\theta_{\ti\tj}=\theta_{[\ti][\tj]}\}_{([\ti],[\tj])\in E_1({\cal G}_{I+\{e\}})},\theta_{\tei\teo}=\chi_{I+\{e\}})
\end{split}
\]
So we obtain an estimate ($\preceq$) by the integral $\I_{I+\{e\}}$ on this region.

In the last case, integrations are performed over the region where parameters in the function are separated from $0$ (are bigger then $\epsilon_0=1$). In this case we integrate a bounded function on the bounded region, so we can estimate this integral by a constant.

The integral $\I_I$ is estimated ($\preceq$) by the sum of the integrals over the regions, so by the sum of $\I_J$ for all $J$ being one step prolongations of $I$ and a constant.
\end{proof}

We have
\begin{lm}\label{cond-d}
Integral \ref{int-simpl} is convergent if
\begin{equation}
 \forall_{I\in \ciag, I\not=\emptyset}\ d(I)>0\ .
\end{equation}
\end{lm}

\begin{proof}
We prove it by induction. Let $\kappa=\max \{|I|\colon I\in\ciag\}=2(|E_0|-1)$. According to lemma \ref{lm-est1} every $\I_I$ for $|I|=\kappa$ is estimated by a constant, so is finite (in fact, this is an integral of constant over the sphere). Suppose we have proved finiteness for all $I\in\ciag$ with $|I|=k$, then for every $J\in \ciag$, $|J|=k-1$ we have by lemma \ref{lm-est1}
\[
 \I_J\preceq 1+\sum_{|I|=k} \I_I,
\]
but the righthand side is finite. By induction we prove that $\I_\emptyset$ is finite.
\end{proof}

\subsection{Graphs}\label{sec-graph}

In this subsection we will prove that $d(I)>0$ for the given sequence $\emptyset\not=I\in\ciag$. 

\begin{df}
The function $\tilde{f}_I(\chi, \{\epsilon_i\}_{i\notin I}, \{\theta_{\ti\tj}\}_{(\ti,\tj)\in E_1({\cal G}_I)})$ is defined as follows
\begin{equation}
 \begin{split}
 \tilde{f}_I(\chi, \{\epsilon_i\}_{i\notin I}&, \{\theta_{\ti\tj}\}_{(\ti,\tj)\in E_1({\cal G}_I)})\\
=f_\emptyset\big(\chi, &\{\epsilon_i\}_{i\notin I},\ \{\epsilon_i=\chi\}_{i\in I},\  \{\theta_{ij}=\theta_{[i][j]}\}_{([i],[j])\in E_1({\cal G}_I)},\ \{\theta_{ij}=\chi\}_{[i]=[j]\ {\rm in\ } E_0({\cal G}_I)}\big)\ . 
 \end{split}
\end{equation}
Explicitly, in function $f_\emptyset$ we put 
\begin{itemize}
\item $\chi$ in place of every $\epsilon_i$ for node $i\in I$,
\item $\chi$ in place of every $\theta_{ij}$ for nodes $i$ and $j$ in the same $\ti\in E_0({\cal G}_I)$ ($[i]=[j]$),
\item $\theta_{[i][j]}$ in place of $\theta_{ij}$ ($[i]$ and $[j]$ are the nodes of ${\cal G}_I$ that contain $i$ and $j$ respectively) if $[i]\not=[j]$.
\end{itemize}
\end{df}

In fact, we have $f_I=\chi^{\rm some\ power}\tilde{f_I}$. Essential part of this property is proved in the following lemma.

\begin{lm}\label{lemma-aux}
The value $d(I)=\Val_\chi f_I$ for $I\in\ciag$ is equal to
\[
 d(I)=\Val_\chi \tilde{f}_I + V+\sum_{\ti\in E_0({\cal G}_I)} 2(N_{\ti}-1)\ ,
\]
where $V$ is the number of nodes belonging to $I$ and $N_{\ti}$ is the number of nodes in $\ti\in E_0({\cal G}_I)$.
\end{lm}

\begin{proof}
In order to obtain $f_I$ we need to put $\chi$ in function $f_\emptyset$ as it is done in $\tilde{f}_I$, but also multiply the function by suitable power of $\chi$.
We prove by induction on the sequences that this additional power factor is $\chi^{V+\sum_{\ti\in E_0({\cal G}_I)} 2(N_{\ti}-1)}$, then $d(I)=\Val_\chi f_I=\Val_\chi \tilde{f}_I+V+\sum_{\ti\in E_0({\cal G}_I)} 2(N_{\ti}-1)$.

If we add a node to $I$ then the additional factor is multiplied by $\chi$, but also $V+\sum_{\ti\in E_0({\cal G}_I)} 2(N_{\ti}-1)$ increases by 1 ($V$ increases by $1$).

If we add a link (say $e$) to $I$ then this factor increases by $\chi^2$, but
we merge two nodes $\tei$ and $\teo$ of ${\cal G}_I$. The merged node is $e_e=[\tei]=[\teo]$ and we have $N_{e_e}=N_{\tei}+N_{\teo}$, so
\[
 2(N_{e_e}-1)=2(N_{\tei}-1)+2(N_{\teo}-1)+2.
\]
We obtain
\[
 V+\sum_{\ti'\in E_0({\cal G}_{I+\{e\}})} 
2(N_{\ti'}-1)=V+
\sum_{\ti\in E_0({\cal G}_I)} 2(N_{\ti}-1)+2\ .
\]
To finish the inductive procedure we need to check that the formula is valid for the empty sequence $I=\emptyset$, but in this case $V+\sum_{\ti\in E_0({\cal G}_I)} 2(N_{\ti}-1)=0$.
\end{proof}

We draw the initial graph with coloured nodes that belong to $I$ and colouring links that  connect two coloured nodes that belong to the same $\tj\in E_0({\cal G}_I)$. Coloured nodes and coloured links form \emph{the coloured graph}. It is associated with the sequence $I$.

\begin{lm}\label{lemma-aux2}
The value $\Val \tilde{f}_I$ can be computed as the sum
\begin{equation}
 \sum_{i\in E_0} d(I)_i+\sum_{e\in E_1} d(I)_e\ ,
\end{equation}
where $d(I)_i$ is equal
\begin{itemize}
 \item $-3$ for coloured node $i$,
\item $0$ for uncoloured node $i$,
\end{itemize}
and $d(I)_e$ is equal
\begin{itemize}
 \item $0$ for coloured link $e$,
\item $0$ for uncoloured link $e$ that connects two uncoloured nodes,
\item $1-\tau$ for uncoloured link $e$ connecting one coloured node with uncoloured one,
\item $2(1-\tau)$ for uncoloured link $e$ that connects two coloured nodes.
\end{itemize}
\end{lm}

\begin{proof}
The value $\Val_\chi \tilde{f}_I$ is the sum of $\Val_\chi$ for factors of $\tilde{f}_I$.

For every node $i\in E_0$ we have factor $\epsilon_i^{-3}$ (if $i\notin I$ and then $\Val_\chi=0$) or $\chi^{-3}$ (if $i\in I$ then $\epsilon_i=\chi$ and $\Val_\chi=-3$). This agrees with the lemma.

For each link we have a factor $\left(\frac{\epsilon_{\ei}\epsilon_{\eo}}{\epsilon_{\ei}^2+\epsilon_{\eo}^2+\theta_{\ei\eo}^2}\right)^{1-\tau}$ (with suitable substitutions by $\chi$).
Let us divide the factor coming from the link into two pieces $\frac{1}{(\epsilon_{\ei}^2+\epsilon_{\eo}^2+\theta_{\ei\eo}^2)^{1-\tau}}$ and $\epsilon_{\ei}^{1-\tau}\epsilon_{\eo}^{1-\tau}$. For the first piece, there are two options:
\begin{itemize}
 \item if nodes $\ei$, $\eo$ are coloured and $[\ei]=[\eo]$ (link is coloured) then it gives $\Val_\chi=-2(1-\tau)$,
\item otherwise (link is uncoloured) it gives an input of $\Val_\chi=0$.
\end{itemize}
The input of factor $\epsilon_{\ei}^{1-\tau}\epsilon_{\eo}^{1-\tau}$ is proportional to the number of coloured nodes among $\ei$, $\eo$ and is $\Val_\chi=0$ (if neither of the two is coloured), $\Val_\chi=1-\tau$ (if only one is coloured), $\Val_\chi=2(1-\tau)$ (if both are coloured).
This agrees with inputs stated in the lemma 
(if link is coloured then it contributes $-2(1-\tau)+2(1-\tau)=0$, otherwise only input of $\epsilon_{\ei}^{1-\tau}\epsilon_{\eo}^{1-\tau}$ counts).
\end{proof}

Now we associate with each connected component $C_\mu$ of the coloured graph a number ${d(C_\mu)}$ in such a way that 
\[
d(I)\geq \sum_\mu d(C_\mu) \ ,
\]
where we sum over all connected components of the coloured graph. We will label components by $C$ with subscript $\mu,\nu,\ldots$. At the end we will prove that for each connected component $d(C_\mu)$ is positive.

\begin{lm}\label{lemma-aux3}
Let for each connected component $C_\mu$ of the coloured graph 
\begin{equation}
d(C_\mu)=-2+E_\mu(1-\tau),
\end{equation}
where $E_\mu$ is the number of uncoloured links that join the connected component $C_\mu$ with the rest of the initial graph. Then
\begin{equation}
 d(I)\geq \sum_\mu d(C_\mu)\ ,
\end{equation}
where we sum over all connected components of the coloured graph.
\end{lm}

\begin{proof}
According to lemma \ref{lemma-aux} we have
$d(I)=\Val_\chi\tilde{f}_I+V+\sum_{\ti\in E_1({\cal G}_I)} 2(N_{\ti}-1)$.

For a given $\ti\in E_0({\cal G}_I)$, the set of coloured nodes that belong to $\ti$ together with all links connecting them (if nonempty) form a connected component of the coloured graph. This is the unique connected component $C_\mu$ in $\ti$ ($E_0(C_\mu)\subset \ti$). In the set~$\ti$ there might be other nodes that do not belong to the coloured graph. Thus, for the node $\ti\in E_0({\cal G}_I)$ holds
\[
 2(N_\mu-1)\leq 2(N_{\ti}-1),
\]
where $N_\mu$ is the number of nodes of the initial graph in connected component $\mu$ and $N_{\ti}$ is the number of nodes in $\ti$. If we sum over all nodes $\ti\in E_0({\cal G}_I)$
\[
 \sum_\mu 2(N_\mu-1)\leq \sum_{\ti\in E_0({\cal G}_I)} 2(N_{\ti}-1)
\]
then the summation on the left-hand side is performed over all connected components of the coloured graph.

Let us compute $\Val_\chi \tilde{f}_I$ (see lemma \ref{lemma-aux2}).
All nodes contribute $\sum_{i\in E_0}d(I)_i=-3V$, where $V$ is the number of coloured nodes.
The contribution of links is as follows. For uncoloured link there are three possibilities
\begin{itemize}
\item it connects two different connected components (say nodes $i\in C_\mu$ and $j\in C_\nu$) and then provides  $\Val_\chi={2(1-\tau)}$ (and is counted twice: in $E_\mu$ and $E_\nu$),
\item it connects a connected component (say node $i\in C_\mu$) with an uncoloured node and then provides $\Val_\chi={1-\tau}$ (and is counted once: in $E_\mu$),
\item it connects two uncoloured nodes and then provides $\Val_\chi=0$ and is not counted in any $E_\mu$.
\end{itemize}
For every coloured link $\Val_\chi=0$ and these links are not counted in any $E_\mu$. So we see, that $\Val_\chi$ contributed by all links is equal to
\[
 \sum_{e\in E_1} d(I)_e=\sum_\mu E_\mu(1-\tau)\ .
\]

According to lemma \ref{lemma-aux} we have
\begin{equation}\label{eq-aux1}
 d(I)=\Val_\chi\tilde{f}_I+V+\sum_{\ti\in E_0({\cal G}_I)} 2(N_{\ti}-1)
\geq -3V+\sum_{\mu} E_\mu(1-\tau)+ V+\sum_\mu 2(N_\mu-1),
\end{equation}
where we sum over all connected components $\mu$ of the coloured graph.
But we have $V=\sum_\mu N_\mu$, so
\[
 -3V+V+\sum_\mu 2(N_\mu-1)=\sum_\mu (-3N_\mu+N_\mu+2N_\mu-2)=\sum_\mu (-2)\ .
\]
Substituting this expression to the equation \ref{eq-aux1} we obtain
\[
 d(I)\geq \sum_\mu \left(-2+E_\mu(1-\tau)\right)\ .
\]
\end{proof}

We note that the coloured graph is not the whole graph (node $0$ does not belong to it).

Suppose first that it is also a non-empty graph. From the property of 3-edge-connectedness (applied to each connected component $C_\mu$ separately) we see that $E_\mu\geq 3$.
Because $E_\mu\tau\leq\frac{1}{2}$ (due to the equation \ref{tau-def}) we have
\[
 -2+E_\mu-E_\mu\tau>0\ .
\]
Application of lemma \ref{lemma-aux3} gives $d(I)>0$ in this case.

If the coloured graph is empty then it means that $I$ consists only of links. We see that $\Val_\chi \tilde{f}_I=0$ and $V=0$ (the number of nodes in $I$). There exists at least one set $\ti\in E_0({\cal G}_I)$ with number of elements $N_{\ti}>1$ (because $I$ is a nonempty sequence), so $d(I)>0$.

Theorem \ref{tech} follows now from lemma \ref{cond-d}.

\section{Conclusions and outlook}

We have proved that all 3-edge-connected spin-networks with the BC or EPRL intertwiners are integrable (corollaries \ref{BC-fin-thm} and \ref{EPRL-fin-thm}). By the Barrett-Crane procedure \cite{BC,BB} one can associate with them a finite evaluation. This allows a definition of the vertex amplitude in the EPRL and BC models for non-simplicial decompositions.
Our proof is valid for a larger class of spin-networks with nodes labelled by elements of $[{\cal F}]$.

Some drawbacks of our proof should also be stressed. In this form it gives very weak estimates on the behaviour of the evaluation with respect to the representations. Namely, as a function of $\rho_{e_1},\ldots, \rho_{e_n}$ in the BC case we have proved only that it is bounded. This is much weaker than estimates obtained in \cite{Rovelli,Cherrington} for special graphs.

\begin{acknowledgments}
Author would like to thank Frank Hellmann,  Marcin Kisielowski  and Jerzy Lewandowski for inspiring discussions on the subject. Paul Hunt, Ludmi{\l}a Janion and especially Micha{\l} Dziendzikowski are thanked for valuable comments on the earlier version of this paper.
This work was supported by grant N N202 287538 of the Polish Ministry of Science and Higher Education (Minsterstwo Nauki i Szkolnictwa Wy\.zszego) and grant \emph{Master} of the Foundation for Polish Science (Fundacja na Rzecz Nauki Polskiej FNP).
\end{acknowledgments}

\appendix
\section{Facts about the invariants}\label{inv-sec}

There is no normalizable invariant vector in the tensor product of nontrivial, irreducible, unitary representations $V_{1}\otimes\cdots\otimes V_{n}$. In fact such a vector would occur in the direct integral decomposition into irreducible representations as an atom in the measure. However, this decomposition \cite{Naimark} consists only (up to measure zero) of representations from the principal and complementary series, where the trivial representation is not present.

In contrast, we will prove that for every $\iota\in{\cal F}(V_{k_1,\rho_1}\otimes\cdots\otimes V_{k_n,\rho_n})$ ($n\geq 3$) the invariant \ref{F-map}
\[
 \int_{SL(2,{\mathbb C})} \rd g\: T_{k_1,\rho_1}(g)\otimes\cdots\otimes T_{k_n,\rho_n}(g) \iota
\]
is defined by duality as a functional on $\cal F$. This means that the integral
\begin{equation}\label{inv-eq}
\int_{SL(2,{\mathbb C})} \rd g \left(\iota',\ T_{k_1,\rho_1}(g)\otimes\cdots\otimes T_{k_n,\rho_n}(g) \iota \right)
\end{equation}
is convergent for every $\iota'\in{\cal F}\left(V_{k_1,\rho_1}^*\otimes\cdots\otimes V_{k_n,\rho_n}^*\right)$.
This expression is the evaluation of the spin-network consisting of two nodes (labelled by $\iota$ and $\iota'$) and $n$ links (labelled by $(k_i,\rho_i)$, $i=1,\ldots, n$) connecting them. According to theorem \ref{fin-thm} (see also \cite{BB}) it is finite if $n\geq 3$.

Now, we prove Fact \ref{thm-label}. To do this it is enough to check that the evaluation is equal zero for a 3-edge-connected spin-network with the label of node $1$  equivalent to $0$ (that is $[\iota_1]=[0]$) and the remaining labels $\{\iota_i\}_{i=0,2,\ldots, |E_0|}$ being simple tensors. General case follows then by linearity. In order to compute the evaluation, we fix $0$ as the node that we are not integrating over. The evaluation does not depend on this choice \cite{BB}. We perform first the integration over $g_1$. The integral is absolutely convergent, so the result does not depend on the order of integration. We obtain $0$ because in the integral we have factor equal to \ref{inv-eq} with $\iota=\iota_1$ (notice that we can reverse orientation of the edges using isomorphisms $V_{k_e,\rho_e}\approx (V_{k_e,\rho_e})^*$).

\end{document}